\documentclass{llncs}
\usepackage[utf8]{inputenc}
\usepackage[english]{babel}
\usepackage{amsmath,amsfonts,amssymb,amsbsy} 

\usepackage{lmodern}
\usepackage[T1]{fontenc}
\usepackage{microtype}

\usepackage{tikz}
\usetikzlibrary{automata}

\newcommand{\tr}[1]{\vphantom{#1}^t #1} 
\newcommand{\vect}[3]{\begin{pmatrix}#1\\#2\\#3\end{pmatrix}}
\newcommand\Ab{\boldsymbol{\ell}} 

\title{Balancedness of Arnoux-Rauzy and Brun words}
\author{Vincent Delecroix\inst{1} \and Tom\'a\v{s} Hejda\inst{2,3} \and Wolfgang Steiner\inst{3}}

\institute{Institut de Math\'ematiques de Jussieu,  CNRS UMR 7586, \\ Universit\'e Paris Diderot -- Paris 7, Case 7012, 75205 Paris Cedex 13, France \\
\email{delecroix@math.jussieu.fr}
\and
Department of Mathematics and Doppler Institute, FNSPE, \\ Czech Technical University in Prague, Czech Republic \\
\email{tohecz@gmail.com}
\and
LIAFA, CNRS UMR 7089, Universit\'e Paris Diderot -- Paris 7, \\ Case 7014, 75205 Paris Cedex 13, France \\ 
\email{steiner@liafa.univ-paris-diderot.fr}}

\begin{document}
\pagestyle{plain}
\maketitle

\begin{abstract}
We study balancedness properties of words given by the Arnoux-Rauzy and Brun multi-dimensional continued fraction algorithms. 
We show that almost all Brun words on $3$~letters and Arnoux-Rauzy words over arbitrary alphabets are finitely balanced; in particular, boundedness of the strong partial quotients implies balancedness. 
On the other hand, we provide examples of unbalanced Brun words on $3$ letters.
\end{abstract}

\section{Introduction}
It is well known that Sturmian words are exactly the $1$-balanced aperiodic words on $2$~letters.
Standard Sturmian words can be characterized in the following way: Each standard Sturmian word $\omega \in \{1,2\}^\mathbb{N}$ is the image of a standard Sturmian word by the substitution $\alpha_1: 1 \mapsto 1,\, 2 \mapsto 12$, or $\alpha_2: 1 \mapsto 21,\, 2 \mapsto 2$; it has thus an `$S$-adic representation' $\omega = \alpha_1^{a_1} \alpha_2^{a_2} \alpha_1^{a_3} \alpha_2^{a_4} \cdots$ (with $S = \{\alpha_1,\alpha_2\}$). 
Moreover, $[0;a_1,a_2,\ldots\,]$ is the continued fraction expansion of~$f_2/f_1$, where $f_i$ denotes the frequency of the letter~$i$ in~$\omega$; e.g., the Fibonacci word is $\omega = \alpha_1 \alpha_2 \alpha_1 \alpha_2 \cdots$, with $[0;1,1,\ldots\,]$ being the golden mean.
For details, we refer to \cite[Chapter~2]{Lothaire02} and \cite[Chapter~6]{PytheasFogg02}.
Since each Sturmian word has the same language as a standard Sturmian word, it is sufficient to study the standard ones for all properties that depend only on the language, such as balancedness. 

Many different generalizations of Sturmian words to larger alphabets can be found in the literature; see e.g.~\cite{Balkova-Pelantova-Starosta10}.
We are interested in words that are provided by multi-dimensional continued fraction algorithms and the corresponding substitutions; see~\cite{Berthe11}.
Since $1$-balancedness is a strong restriction~\cite{Hubert00,Vuillon03}, we are interested in finite balancedness of words given by the Arnoux-Rauzy and Brun continued fraction algorithms; see
Sections~\ref{sec:notation} and~\ref{sec:ar_and_brun_words} for precise definitions.

The prototype of an Arnoux-Rauzy word is the Tribonacci word, which is $2$-balanced~\cite{Richomme-Saari-Zamboni10}.
However, we know from \cite{Cassaigne-Ferenczi-Zamboni00} that there are Arnoux-Rauzy words (on $3$~letters) that are not finitely balanced; see also~\cite{Cassaigne-Ferenczi-Messaoudi08}. 
In~\cite{Berthe-Cassaigne-Steiner}, it was shown that Arnoux-Rauzy words are finitely balanced if the `weak partial quotients' are bounded, and that a large class of Arnoux-Rauzy words are $2$-balanced. 
Here, we show that the set of finitely balanced Arnoux-Rauzy words has full measure (with respect to a suitably chosen measure on Arnoux-Rauzy words), and contains the words with bounded `strong partial quotients' (in arbitrary dimension). 
Note however that, for $d \geq 3$, Arnoux-Rauzy words are defined only for a set of slopes of zero Lebesgue measure that form the the so-called Rauzy gasket~\cite{Arnoux-Starosta13}.

The Brun algorithm has the advantage over Arnoux-Rauzy that it is defined for all directions in~$\mathbb{R}_+^d$. 
To our knowledge, the balancedness of words associated to the Brun algorithm has not been studied yet. 
We show that almost all Brun words on $3$~letters are finitely balanced; in particular, this holds for words with bounded `strong partial quotients'. 
We also exhibit Brun words (on $3$~letters) that are not finitely balanced.
Note that, for fixed points of substitutions, an exact criterion for balancedness is provided by~\cite{Adamczewski03}.

\section{Notation} \label{sec:notation}
Let $\mathcal{A} = \{1,2,\ldots,d\}$ be a finite alphabet and $\mathcal{A}^*$ be the free monoid over~$\mathcal{A}$ (with the concatenation as product).
Let $|w|$ be the length of a word $w \in \mathcal{A}^*$ and $|w|_j$ the number of occurrences of the letter $j \in \mathcal{A}$ in~$w$.
A~pair of words $u, v \in \mathcal{A}^*$ with $|u| = |v|$, is \emph{$C$-balanced} if 
\[
-C \le |u|_j - |v|_j \le C \quad \mbox{for all}\ j \in \mathcal{A}.
\]
A~\emph{factor} of an infinite word $\omega  = (\omega_n)_{n\in\mathbb{N}} \in \mathcal{A}^\mathbb{N}$ is a finite word of the form $\omega_{[k,\ell)} = \omega_k \omega_{k+1} \cdots \omega_{\ell-1}$.
An infinite word~$\omega$ is \emph{$C$-balanced} if each pair of factors $u, v$ of $\omega$ with $|u| = |v|$ is $C$-balanced; $\omega$~is \emph{finitely balanced} if it is $C$-balanced for some $C \in \mathbb{N}$.
The \emph{balance} of an infinite word $\omega$ is the smallest number $B(\omega)$ such that $\omega$ is $B(\omega)$-balanced, with $B(\omega) = \infty$ if $\omega$ is not finitely balanced.

The \emph{frequency}~$f_i$ of a letter $i \in \mathcal{A}$ in $\omega = (\omega_n)_{n\in\mathbb{N}} \in \mathcal{A}^\mathbb{N}$ is $\lim_{n\to\infty} |\omega_{[0,n)}|_i/n$, if the limit exists.
It is easy to see that the frequency of each letter exists when $\omega$ is finitely balanced (see \cite{Berthe-Tijdeman:02}).

A~\emph{substitution} $\sigma$ over~$\mathcal{A}$ is an endomorphism of~$\mathcal{A}^*$.
Its \emph{incidence matrix} is the square matrix $M_\sigma = (|\sigma(j)|_i)_{i,j\in\mathcal{A}} \in \mathbb{N}^{d\times d}$ (with $\mathbb{N} = \{0,1,2,\ldots\}$).
The map
\[
\Ab:\ \mathcal{A}^* \to\mathbb{N}^d, \ w \mapsto \tr{(|w|_{1},|w|_2,\ldots, |w|_{d})}
\]
is called the \emph{abelianization map}.
Note that $\Ab(\sigma(w)) = M_\sigma \Ab(w)$ for all $w\in \mathcal{A}^*$.

Let $(\sigma_n)_{n\in\mathbb{N}}$ be a sequence of substitutions over the alphabet~$\mathcal{A}$.
To keep notation concise, we set $M_n = M_{\sigma_n}$ for $n \in \mathbb{N}$ and denote products of consecutive substitutions and their incidence matrices by $\sigma_{[k,\ell)}=\sigma_k \sigma_{k+1} \cdots \sigma_{\ell-1}$ and $M_{[k,\ell)}=M_k M_{k+1}\cdots M_{\ell-1}$ respectively.
A~word $\omega \in \mathcal{A}^\mathbb{N}$ is a \emph{limit word} of $(\sigma_n)_{n\in\mathbb{N}}$ if there is a sequence $(\omega^{(n)})_{n\in\mathbb{N}}$ with 
\[
\omega^{(0)} = \omega, \quad \omega^{(n)} = \sigma_n\big(\omega^{(n+1)}\big) \quad \mbox{for all}\ n \in \mathbb{N},
\]
where the substitutions $\sigma_n$ are extended naturally to infinite words.
The word~$\omega$ is called an \emph{$S$-adic word} with \emph{directive sequence} $(\sigma_n)_{n\in\mathbb{N}}$ and $S = \{\sigma_n: n \in \mathbb{N}\}$. 
 
Given a directive sequence $(\sigma_n)_{n \in \mathbb{N}}$, we can define different generalizations of partial quotients. 
The sequence of \emph{weak partial quotients} is the sequence of positive integers $(a_n)_{n \in \mathbb{N}}$ such that $\sigma_0 \sigma_1 \cdots = \sigma_{A_0}^{a_0} \sigma_{A_1}^{a_1} \cdots$, with $A_n = \sum_{k=0}^{n-1} a_k$ and $\sigma_{A_n} = \cdots = \sigma_{A_{n+1}-1} \ne \sigma_{A_{n+1}}$ for all $n \in \mathbb{N}$. 
The notion of strong partial quotients refers to the time we need to reach a positive (or at least primitive) matrix in the product of incidence matrices. 
A~good precise definition of them probably depends on~$S$ and the intended use, but properties like being bounded should hold simultaneously for all suitable definitions.
In this paper, we say that the \emph{strong partial quotients are bounded by~$h$} if $M_{[n,n+h)}$ is primitive for all $n \in \mathbb{N}$.

\section{Arnoux-Rauzy and Brun words} \label{sec:ar_and_brun_words}
We are interested in this paper in two $S$-adic systems that arise naturally from multi-dimensional continued fraction algorithms. 
The set of \emph{Arnoux-Rauzy substitutions} over $d$~letters is $S_\mathrm{AR} = \{\alpha_i: i \in \mathcal{A}\}$ with
\[
	\alpha_i:\ i \mapsto i,\ j \mapsto ij\ \mbox{for}\ j \in \mathcal{A} \setminus \{i\}.
\]
For each directive sequence $(\sigma_n)_{n\in\mathbb{N}} = (\alpha_{i_n})_{n\in\mathbb{N}} \in S_\mathrm{AR}^\mathbb{N}$, the words $\sigma_{[0,n)}(i_n)$ are nested prefixes of the limit word~$\omega$. 
If the directive sequence contains infinitely many occurrences of each substitution~$\alpha_i$, $i \in \mathcal{A}$, the unique limit word~$\omega$ is called a \emph{standard Arnoux-Rauzy word}. 
Any word that has the same language (and thus the same balancedness properties) as a standard Arnoux-Rauzy word is called Arnoux-Rauzy word.
The \emph{Tribonacci word} is the Arnoux-Rauzy word on 3-letters with periodic
directive sequence $\alpha_1 \alpha_2 \alpha_3 \alpha_1 \alpha_2 \alpha_3 \cdots$; it satisfies
\vspace{-.25ex}
\[
\omega = 121312112131212131211213121312112131212131211213121121\cdots
\vspace{-.25ex}
\]
and is known to be 2-balanced~\cite{Richomme-Saari-Zamboni10}.

A~set of \emph{Brun substitutions} was defined in~\cite{Berthe-Fernique11} to provide a connection between stepped planes and the Brun algorithm. 
Here, we consider the set of substitutions $S_\mathrm{Br} = \{\beta_{ij}: i \in \mathcal{A},\, j \in \mathcal{A} \setminus \{i\}\}$ over $d$~letters, with
\[
\beta_{ij}:\ j \mapsto ij,\ k\mapsto k\ \text{for}\ k \in \mathcal{A} \setminus \{j\},
\]
that corresponds to the additive version of this algorithm.
An $S_\mathrm{Br}$-adic word is called a \emph{Brun word} if its directive sequence $(\sigma_n)_{n\in\mathbb{N}}$ satisfies
\begin{multline}\label{eq:Brun:admis}
\sigma_n\sigma_{n+1} \in \big\{\beta_{ij}\beta_{ij}: i \in \mathcal{A},\, j \in \mathcal{A} \setminus \{i\}\big\} \\ \cup \big\{\beta_{ij}\beta_{jk}: i \in \mathcal{A},\, j \in \mathcal{A} \setminus \{i\},\, k \in \mathcal{A} \setminus \{j\}\big\} \quad \mbox{for all}\ n \in \mathbb{N}
\end{multline}
and for each $i \in \mathcal{A}$ there is $j \in \mathcal{A}$ such that $\beta_{ij}$ occurs infinitely often in $(\sigma_n)_{n\in\mathbb{N}}$. 
E.g., the Brun word with periodic directive sequence $\beta_{12} \beta_{23} \beta_{31} \beta_{12} \beta_{23} \beta_{31} \cdots$ is
\vspace{-.25ex}
\[
\omega = 1231121231231123112123112123123112123123112311212312\cdots.
\vspace{-.25ex}
\]

Recall that the Brun algorithm~\cite{Brun58} subtracts at each step the second largest coordinate from the largest coordinate. 
It is given by the transformations
\[
T_{ij}:\ D_{ij} \to \mathbb{R}_+^d,\ \mathbf{f} \mapsto M_{\beta_{ij}}^{-1}\, \mathbf{f} / \|M_{\beta_{ij}}^{-1}\, \mathbf{f}\|_1, 
\] 
where $D_{ij} \subset \mathbb{R}_+^d$ is the set of vectors $\mathbf{f} =\tr{(f_1,\dots,f_d)}$ such that $f_i\geq f_j$ are the two largest components  of $\mathbf{f}$; here, $\mathbb{R}_+ = [0,\infty)$.
A sequence $(i_0,j_0)(i_1,j_1)(i_2,j_2)\dotsm$ is a Brun representation of $\mathbf{f} \in\mathbb R_+^d$ if
\[
	T_{i_{k-1}j_{k-1}} \cdots T_{i_1j_1} T_{i_0j_0}(\mathbf{f}) \in D_{i_kj_k}
\quad\text{for all}\quad
	k\in\mathbb{N}
.\]
Given a Brun word $\omega$ with directive sequence $\beta_{i_0j_0}\beta_{i_1j_1}\beta_{i_2j_2}\dotsm$,
 we get that $(i_0,j_0)(i_1,j_1)(i_2,j_2)\dotsm$ is a Brun representation of the vector of frequencies of $\omega$.

In the Arnoux-Rauzy algorithm, all but one coordinates are subtracted from the largest coordinate, which is assumed to be larger than the sum of the other coordinates.
Here, we have transformations $T_i: D_i \to \mathbb{R}_+^d$ with $D_i \subset \mathbb{R}_+^d$ being the set of vectors $\mathbf{f} =\tr{(f_1,\dots,f_d)}$ such that $f_i \geq \sum_{j\in\mathcal{A}\setminus\{i\}} f_j$.

The following two lemmas translate the fact that these two algorithms converge and show that the frequency vector $\mathbf{f} = \mathbf{f}(\omega)$ of the limit word of a directive sequence $(\sigma_n)_{n\in\mathbb{N}}$ is given by the limit cone
\begin{equation} \label{e:limitcone}
\mathbb{R}_+\, \mathbf{f} = \bigcap_{n\in\mathbb{N}} M_{[0,n)}\, \mathbb{R}_+^d.
\end{equation}
Moreover, because of the relation with the continued fraction algorithms, two distinct standard Arnoux-Rauzy words and two distinct standard Brun words respectively have different frequency vectors.

\begin{lemma} \label{l:frequencyBrun}
Each Brun word on $3$~letters has letter frequencies.
\end{lemma}

\begin{proof}
Let $(\sigma_n)_{n\in\mathbb{N}} \in S_\mathrm{Br}^\mathbb{N}$ be a directive sequence of a Brun word on $3$ letters, $M_n$~the associated incidence matrices, and $\mathbf{f} \in \bigcap_{n\in\mathbb{N}} M_{[0,n)}\, \mathbb{R}_+^d$.
From~\cite{Avila-Delecroix}, we know that there is a sequence of matrices $(\tilde{M}_n)_{n\in\mathbb{N}}$ such that $\|\tr{(\tilde{M}_{[0,n)})}\|_\infty \le 1$,
\begin{equation} \label{e:tMtilde}
\tr{(\tilde{M}_{[0,n)})}\, \mathbf{x} = \tr{(M_{[0,n)})}\, \mathbf{x} \quad \mbox{for all} \quad  \mathbf{x} \in \mathbf{f}^\bot,\ n \in \mathbb{N},
\end{equation}
where $\mathbf{f}^\bot$ denotes the hyperplane orthogonal to~$\mathbf{f}$; see also Section~\ref{sec:contr-brun-matr}.

For each $\mathbf{v} \in \mathbf{f}^\bot$ with $\|\mathbf{v}\|_\infty = 1$, we have 
\[
\bigg| \bigg\langle \mathbf{v}, \frac{M_{[0,n)}\,\mathbf{e}_i}{\|M_{[0,n)}\,\mathbf{e}_i\|_1} \bigg\rangle \bigg| = \bigg| \bigg\langle \tr{(\tilde{M}_{[0,n)})}\, \mathbf{v}, \frac{\mathbf{e}_i}{\|M_{[0,n)}\,\mathbf{e}_i\|_1} \bigg\rangle \bigg| \le \frac{1}{\|M_{[0,n)}\,\mathbf{e}_i\|_1} 
\]
for each unit vector~$\mathbf{e}_i$, $i \in \mathcal{A}$.
Since $\min_{i\in\mathcal{A}}\|M_{[0,n)}\,\mathbf{e}_i\|_1 \to \infty$ for each directive sequence of a Brun word, the cone $M_{[0,n)}\, \mathbb{R}_+^d$ tends to the line $\mathbb{R}_+\, \mathbf{f}$, i.e., \eqref{e:limitcone} holds. 
From~\eqref{e:limitcone}, it is standard to prove that $\mathbf{f}$ is the frequency vector of~$\omega$; see~\cite{Berthe-Delecroix}. \qed
\end{proof}

The proof of Lemma~\ref{l:frequencyBrun} could be adapted to Arnoux-Rauzy words because the incidence matrix of an Arnoux-Rauzy substitution is similar to a matrix given by the fully subtractive algorithm, which was studied in~\cite{Avila-Delecroix}.
However, we prefer using the results of~\cite{Avila-Delecroix} in a different way in the proof of the following lemma.

\begin{lemma} \label{l:frequencyAR}
Each Arnoux-Rauzy word (on $d \ge 2$ letters) has letter frequencies.
\end{lemma}

\begin{proof}
Let $(\sigma_n)_{n\in\mathbb{N}} = (\alpha_{i_n})_{n\in\mathbb{N}}$ be the directive sequence of an Arnoux-Rauzy word and $\mathbf{f} \in \bigcap_{n\in\mathbb{N}} M_{[0,n)}\, \mathbb{R}_+^d$ with $\|\mathbf{f}\|_1 = 1$. 
We know from the results on the fully subtractive algorithm in~\cite{Avila-Delecroix} that there is a sequence of matrices $(\tilde{M}_n)_{n\in\mathbb{N}}$ such that $\|\tilde{M}_{[0,n)}\|_\infty \le 1$ and
\begin{equation} \label{e:Mtilde}
\tilde{M}_{[0,n)}\, \mathbf{x} = M_{[0,n)}\, \mathbf{x} \quad \mbox{for all} \quad  \mathbf{x} \in (M_{[0,n)})^{-1} \mathbf{1}^\bot,\ n \in \mathbb{N},
\end{equation}
where $\mathbf{1} = \tr{(1,\ldots,1)}$; see also Section~\ref{sec:contr-arno-rauzy}.

Denote by $\pi_n$ be the projection along $(M_{[0,n)})^{-1} \mathbf{f}$ onto $(M_{[0,n)})^{-1} \mathbf{1}^\bot$.
Then
\[
\|\pi_0\,M_{[0,n)}\,\Ab(i_n)\|_\infty = \|M_{[0,n)}\,\pi_n\,\Ab(i_n)\|_\infty = \|\tilde{M}_{[0,n)}\,\pi_n\,\Ab(i_n)\|_\infty \le \|\pi_n\,\Ab(i_n)\|_\infty.
\]
Since $(M_{[0,n)})^{-1} \mathbf{f} \in M_n\, \mathbb{R}_+^d$, the $i_n$-th coordinate of $(M_{[0,n)})^{-1} \mathbf{f}$ is larger than or equal to (the sum of) the other coordinates, thus $\|\pi_n\,\Ab(i_n)\|_\infty \le \|\Ab(i_n)\|_\infty = 1$.

Following Dumont and Thomas~\cite{Dumont-Thomas89}, each prefix $\omega_{[0,k)}$ of~$\omega$ can be written as
\[
\omega_{[0,k)} = \sigma_{[0,m-1)}(p_{m-1})\, \sigma_{[0,m-2)}(p_{m-2})\, \cdots\, \sigma_{[0,1)}(p_1)\, p_0,
\]
with a sequence of words $(p_n)_{0\le n<m}$ defined in the following way. 
Let $m = m(k) \in \mathbb{N}$ be minimal such that $|\sigma_{[0,m)}(i_m)| > k$.
Then there is a unique prefix $p_{m-1}$ of $\sigma_{m-1}(i_m)$ such that 
\[
|\sigma_{[0,m-1)}(p_{m-1})| \le k < |\sigma_{[0,m-1)}(p_{m-1} a_{m-1})|,
\]
with $a_{m-1} \in \mathcal{A}$ being the letter following~$p_{m-1}$ in~$\sigma_{m-1}(i_m)$.
Inductively, we obtain for $0 \le n < m$ unique $p_n \in \mathcal{A}^*$ and $a_n \in \mathcal{A}$  such that 
\[
|\sigma_{[0,n)}(p_n)| \le k - \sum_{j=n+1}^{m-1}|\sigma_{[0,j)}(p_j)| < |\sigma_{[0,n)}(p_n a_n)|
\] 
and $p_n a_n$ is a prefix of $|\sigma_n(a_{n+1})|$, with $a_m = i_m$.
We thus have
\[
\Ab(\omega_{[0,k)}) = \sum_{n=0}^{m-1} \Ab\big(\sigma_{[0,n)}(p_n)\big) = \sum_{n=0}^{m-1} \pi_0\, M_{[0,n)}\, \Ab(p_n).
\]
By the definition of the Arnoux-Rauzy substitutions, $p_n$~is either empty or equal to~$i_n$, thus
\[
\big\| \pi_0\, \Ab(\omega_{[0,k)}) \big\|_\infty \le  \sum_{n=0}^{m-1} \big\| \pi_0\, M_{[0,n)}\, \Ab(i_n) \big\|_\infty \le m.
\]
Since $m(k) / \|\Ab(\omega_{[0,k)})\|_\infty \to 0$ as $k \to \infty$, the direction of $\Ab(\omega_{[0,k)})$ converges to that of~$\mathbf{f}$, thus $\mathbf{f}$~is the frequency vector of~$\omega$. \qed
\end{proof}

\section{Discrepancy and balancedness}
Let $\omega$ be an infinite word with frequency vector $\mathbf{f} = \tr{(f_1, f_2, \ldots, f_d)}$, and denote by $\pi$ be the projection along~$\mathbf{f}$ onto~$\mathbf{1}^\bot$.
It is easily written down in coordinates: 
\[
\pi\, \Ab(w) = \tr{(|w|_1 - |w|\, f_1,\, |w|_2 - |w|\,f_2,\, \ldots,\, |w|_d - |w|\, f_d)}.
\]
Note that the so called Rauzy fractal is the closure of $\{\pi\, \Ab(\omega_{[0,n)}): n \in \mathbb{N}\}$, which is the projection of the vertices of the broken line associated with~$\omega$.

More generally, for a function $\phi: \mathcal{A} \rightarrow \mathbb{R}$, we consider the Birkhoff sums
\[
S_n(\phi,\omega) = \phi(\omega_0) + \phi(\omega_1) + \cdots + \phi(\omega_{n-1}).
\]
Remark that, if $\chi_i$ denotes the characteristic function of a letter $i \in \mathcal{A}$, then $S_n(\chi_i,\omega) = |\omega_{[0,n)}|_i$, and the coordinates of $\pi\, \Ab(\omega_{[0,n)})$ are $S_n(\chi_i - f_i, \omega)$.
The $\phi$-\emph{discrepancy} of $\omega$ is $\Delta(\phi,\omega) = \sup_{n\in\mathbb{N}} |S_n(\phi,\omega)|$.
We set
\[
\Delta(\omega) = \max_{i\in\mathcal{A}} \Delta(\chi_i - f_i, \omega) = \sup_{n\in\mathbb{N}} \| \pi\, \Ab(\omega_{[0,n)}) \|_\infty,
\]
and say that $\omega$ has \emph{finite discrepancy} if $\Delta(\omega) < \infty$.
The following result from~\cite[Proposition~7 and Remark~8]{Adamczewski03}  establishes a link between balance and discrepancy. 

\begin{lemma} \label{l:bounded}
We have $\Delta(\omega) \le B(\omega) \le 4 \Delta(\omega)$.
\end{lemma}

For many words, balancedness can be shown using the following proposition.

\begin{proposition} \label{p:bounded3}
Let $\omega$ be an Arnoux-Rauzy or Brun word with directive sequence $(\sigma_n)_{n\in\mathbb{N}}$.
For each sequence of matrices $(\tilde{M}_n)_{n\in\mathbb{N}}$ satisfying~\eqref{e:tMtilde}, we have
\[
\Delta(\omega) \le \sum_{n=0}^\infty \big\|\tr{(\tilde{M}_{[0,n)})}\big\|_\infty.
\]
For each sequence of matrices $(\tilde{M}_n)_{n\in\mathbb{N}}$ satisfying~\eqref{e:Mtilde}, we have
\[
\Delta(\omega) \le \sum_{n=0}^\infty \big\|\tilde{M}_{[0,n)}\big\|_\infty.
\]
\end{proposition}

\begin{proof}
The first statement follows from 
\begin{align*}
\Delta(\chi_i - f_i, \omega) & = \sup_{k\in\mathbb{N}} \big| \big\langle \mathbf{e}_i - f_i\, \mathbf{1}, \Ab(\omega_{[0,k)}) \big\rangle \big| \le \sum_{n=0}^\infty \big| \big\langle \mathbf{e}_i - f_i\, \mathbf{1}, M_{[0,n)}\, \Ab(i_n) \big\rangle \big| \\
& = \sum_{n=0}^\infty \big| \big\langle \tr{(\tilde{M}_{[0,n)})}\, (\mathbf{e}_i - f_i\, \mathbf{1}), \Ab(i_n) \big\rangle \big| \le \sum_{n=0}^\infty \big\|\tr{(\tilde{M}_{[0,n)})}\big\|_\infty,
\end{align*}
where we have used the Dumont-Thomas representations in the first inequality (see the proof of Lemma~\ref{l:frequencyAR}), the fact that $\mathbf{e}_i - f_i\, \mathbf{1} \in \mathbf{f}^\bot$ in the second equality and $\|\mathbf{e}_i - f_i\, \mathbf{1}\|_\infty \le 1$ in the last inequality.
The proof of the second statement runs along the lines of the proof of Lemma~\ref{l:frequencyAR}, where we can replace $\alpha_{i_n}$ by~$\beta_{i_n j_n}$. \qed
\end{proof}

\section{Contractivity of Arnoux-Rauzy matrices} \label{sec:contr-arno-rauzy}
Now we study the contractivity of Arnoux-Rauzy matrices on certain hyper\-planes, quantifying the approach in~\cite{Avila-Delecroix}.
For a directive sequence $(\sigma_n)_{n\in\mathbb{N}}$, let
\[
\mathbf{v}^{(n)} = \tr{\big(v^{(n)}_1,v^{(n)}_2,\ldots,v^{(n)}_d\big)} = \frac{\tr{(M_{[0,n)})}\, \mathbf{1}}{\|\tr{(M_{[0,n)})}\, \mathbf{1}\|_1} \qquad (n \in \mathbb{N}). 
\]

\begin{lemma} \label{l:AR1}
Let $\omega$ be an Arnoux-Rauzy  word with directive sequence $(\alpha_{i_n})_{n\in\mathbb{N}}$.
Then $\|\mathbf{v}^{(n)}\|_\infty < \frac{1}{d-1}$ for all $n \in \mathbb{N}$. 
If moreover $\{i_n,i_{n+1},\ldots,i_{n+h-1}\} = \mathcal{A}$, $h \in \mathbb{N}$, then $\|\mathbf{v}^{(n+h)}\|_\infty < \frac{2^h-1}{2^h(d-1)}$. 
\end{lemma}

\begin{proof}
First note that $\|\mathbf{v}^{(0)}\|_\infty = \frac{1}{d}$.
Assume now that $\|\mathbf{v}^{(n)}\|_\infty < \frac{1}{d-1}$, and let w.l.o.g.\ $i_n = 1$. 
Then the simplex $\big\{\mathbf{x} \in \mathbb{R}^d: \|\mathbf{x}\|_1 = 1,\, \|\mathbf{x}\|_\infty \le \frac{1}{d-1}\big\}$ is mapped by $\tr{\!M_n}$ (after normalizing) to the simplex spanned by $\tr{\big(0,\frac{1}{d-1},\ldots,\frac{1}{d-1}\big)}$, $\tr{\big(\frac{1}{2(d-1)},\frac{1}{2(d-1)},\frac{1}{d-1},\ldots,\frac{1}{d-1}\big)}$, \dots, $\tr{\big(\frac{1}{2(d-1)},\frac{1}{d-1},\ldots,\frac{1}{d-1},\frac{1}{2(d-1)}\big)}$.
This shows that $\|\mathbf{v}^{(n+1)}\|_\infty < \frac{1}{d-1}$ and that $v^{(n+1)}_{i_n} < \frac{1}{2(d-1)}$. 
Similar considerations show that $v^{(n+h)}_{i_n} < \frac{2^h-1}{2^h(d-1)}$ for all $h \ge 1$.
If $\{i_n,i_{n+1},\ldots,i_{n+h-1}\} = \mathcal{A}$, then we obtain that $v_j^{(n+h)} < \frac{2^h-1}{2^h(d-1)}$ for all $j \in \mathcal{A}$. \qed
\end{proof}

\begin{lemma} \label{l:AR2}
Let $\omega$ be an Arnoux-Rauzy  word with directive sequence $(\alpha_{i_n})_{n\in\mathbb{N}}$.
Then there is a sequence of matrices $(\tilde{M}_n)_{n\in\mathbb{N}}$ satisfying~\eqref{e:Mtilde} with
\[
\big\|\tr{\!\tilde{M}_n}\, \mathbf{e}_{i_n}\big\|_1 = d - \|\mathbf{v}^{(n+1)}\|_\infty^{-1} < 1 
\]
and $\tr{\!\tilde{M}_n}\, \mathbf{e}_j = \mathbf{e}_j$ for all $j \in \mathcal{A} \setminus \{i_n\}$, $n \in \mathbb{N}$.
\end{lemma}

\begin{proof}
For each $n \in \mathbb{N}$, let $\tilde{M}_n$ be the matrix with $i_n$-th row $\tr{\mathbf{1}} - \frac{\tr{\mathbf{v}}^{(n+1)}}{\|\mathbf{v}^{(n+1)}\|_\infty}$ and $j$-th row $\tr{\mathbf{e}_j}$ for all $j \in \mathcal{A} \setminus \{i_n\}$. 
Then $\big\|\tr{\!\tilde{M}_n}\, \mathbf{e}_{i_n}\big\|_1 = d - \|\mathbf{v}^{(n+1)}\|_\infty^{-1}$ and $\tr{\!\tilde{M}_n}\, \mathbf{e}_j = \mathbf{e}_j$ for all $j \in \mathcal{A} \setminus \{i_n\}$, with $\|\mathbf{v}^{(n+1)}\|_\infty^{-1} > d-1$ by Lemma~\ref{l:AR1}.
Since adding a multiple of~$\tr{\mathbf{v}}^{(n+1)}$ to a row of~$M_n$ does not change $M_n\, \mathbf{x}$ for $\mathbf{x} \in (\mathbf{v}^{(n+1)})^\bot$, we have $\tilde{M}_n\, \mathbf{x} = M_n\, \mathbf{x}$ for all $\mathbf{x} \in (\mathbf{v}^{(n+1)})^\bot$.
Using that $M_n\, (\mathbf{v}^{(n+1)})^\bot = (\mathbf{v}^{(n)})^\bot$, we obtain inductively that \eqref{e:Mtilde} holds, which proves the lemma. \qed
\end{proof}

\begin{lemma} \label{l:AR3}
Let $\omega$ be an Arnoux-Rauzy  word with directive sequence $(\alpha_{i_n})_{n\in\mathbb{N}}$ and $(\tilde{M}_n)_{n\in\mathbb{N}}$ as in Lemma~\ref{l:AR2}.
If $\{i_k,i_{k+1},\ldots,i_{\ell-1}\} = \mathcal{A}$ and there is $h \in \mathbb{N}$ such that $\{i_{n-h+1},i_{n-h+2},\ldots,i_n\} = \mathcal{A}$ for all $n \in [k,\ell)$, then
\[
\big\|\tilde{M}_{[k,\ell)}\big\|_\infty < \frac{2^h-d}{2^h-1}.
\]
\end{lemma}

\begin{proof}
Let $j \in \mathcal{A}$ and let $m \in [k,\ell)$ be minimal such that $i_m = j$.
Then
\[
\big\|\tr{(\tilde{M}_{[k,\ell)})}\, \mathbf{e}_j\big\|_1 \le \big\|\tr{(\tilde{M}_{[m+1,\ell)})}\big\|_1\, \big\|\tr{(\tilde{M}_{[k,m+1)})}\, \mathbf{e}_j\big\|_1 \le \big\|\tr{\!\tilde{M}_m}\, \mathbf{e}_j\big\|_1 < \frac{2^h-d}{2^h-1},
\]
where we have used that, for all $n \in [k,\ell)$ by Lemmas~\ref{l:AR1} and~\ref{l:AR2}, $\|\tr{\!\tilde{M}_n}\| \le 1$, $\tr{\!\tilde{M}_n}\, \mathbf{e}_j = \mathbf{e}_j$ for all $j \in \mathcal{A} \setminus \{i_n\}$, and $\|\tr{\!\tilde{M}_n}\, \mathbf{e}_{i_n}\|_1 < \frac{2^h-d}{2^h-1}$.
This shows that $\|\tilde{M}_{[k,\ell)}\|_\infty = \|\tr{(\tilde{M}_{[k,\ell)})}\|_1 < \frac{2^h-d}{2^h-1}$. \qed
\end{proof}

\begin{theorem} \label{t:balanced}
Let $h \in \mathbb{N}$. 
There is a constant $C(h)$ such that each Arnoux-Rauzy word with strong partial quotients bounded by~$h$, i.e., with directive sequence $(\alpha_{i_n})_{n\in\mathbb{N}}$ satisfying $\{i_n,\ldots,i_{n+h-1}\} = \mathcal{A}$ for all $n \in \mathbb{N}$, is $C(h)$-balanced.
\end{theorem}

\begin{proof}
By Lemma~\ref{l:AR3}, there is a sequence $(\tilde{M}_n)_{n\in\mathbb{N}}$ satisfying~\eqref{e:Mtilde} such that $\|\tilde{M}_{[n,n+h)}\|_\infty < \frac{2^h-d}{2^h-1}$ for all $n \ge h-1$, thus $\|\tilde{M}_{[0,n)}\|_\infty = \mathcal{O}\big(\big(\frac{2^h-d}{2^h-1}\big)^{n/h}\big)$, hence $\sum_{n=0}^\infty \|\tilde{M}_{[0,n)}\|_\infty$ is bounded. 
Lemma~\ref{l:bounded} and Proposition~\ref{p:bounded3} conclude the proof. \qed
\end{proof}

\section{Contractivity of $3$-dimensional Brun matrices} \label{sec:contr-brun-matr}
For Brun words (over $3$~letters), we follow a similar strategy as for Arnoux-Rauzy words. 
For a Brun word~$\omega$ with directive sequence $(\psi_{i_n,j_n})_{n\in\mathbb{N}}$, let $(k_n)_{n\in\mathbb{N}}$ be the sequence of letters defined by $\{i_n,j_n,k_n\} = \mathcal{A}$, and let
\[
\mathbf{f}^{(n)} = \tr{\big(f^{(n)}_1,f^{(n)}_2,\ldots,f^{(n)}_d\big)} = \frac{(M_{[0,n)})^{-1}\mathbf{f}}{\|(M_{[0,n)})^{-1}\mathbf{f}\|_1}
\]
be the frequency vector of $\omega^{(n)}$.
Moreover, let $(F_n)_{n\in\mathbb{N}}$ be the sequence of Fibonacci numbers defined by $F_0 = 1$, $F_1 = 2$, $F_n = F_{n-1} + F_{n-2}$ for all $n \ge 2$.

\begin{lemma}
Let $\omega$ be a Brun word over $3$ letters with directive sequence $(\beta_{i_n,j_n})_{n\in\mathbb{N}}$.
Then $f_{i_n}^{(n-h)} \ge \frac{1}{F_{h+1}+1}$ for all $h \le n$. 
If $\{i_n,i_{n+1},\ldots,i_{n+h-1}\} = \mathcal{A}$ for all $n \in \mathbb{N}$, then we have $(f_{j_n}^{(n)}-f_{k_n}^{(n)})/f_{i_n}^{(n)} \ge \frac{1}{F_h}$ for all $n \in \mathbb{N}$. 
\end{lemma}

\begin{proof}
Since $f_{i_n}^{(n)} \ge f_j^{(n)}$ for all $j \in \mathcal{A}$, we have $f_{i_n}^{(n)} \ge 1/3$, and it is easily checked that the minimum for $f_{i_n}^{(n-h)}$ is attained when $\mathbf{f}^{(n)} = (1/3,1/3,1/3)$ and $i_{n-h} \cdots i_{n-1}$ is an alternating sequence of $j_n$ and~$k_n$. 
In this case, we have $f^{(n-h)}_{i_n} = 1$, $f^{(n-h)}_{i_{n-h}} = F_h$, and $f^{(n-h)}_{i_{n-h+1}} = F_{h-1}$, thus $f^{(n-h)}_{i_n} = \frac{1}{F_{h+1}+1}$.

Let now, w.l.o.g.\ $i_n = 1$, $j_n=2$, and assume that $\{1,3\} \subset \{i_{n+2},\ldots,i_{n+h}\}$.
Then $\mathbf{f}^{(n+1)}$ lies in the quadrangle with corners $\big(\frac{1}{3}, \frac{1}{3}, \frac{1}{3}\big)$, $\big(\frac{F_h}{2(F_h+1)}, \frac{F_h}{2(F_h+1)}, \frac{1}{F_h+1}\big)$, $\big(\frac{1}{F_h+1}, \frac{F_h-1}{F_h+1}, \frac{1}{F_h+1}\big)$, and $\big(\frac{1}{F_h+1}, \frac{F_h}{2(F_h+1)}, \frac{F_h}{2(F_h+1)}\big)$.
Therefore, $\mathbf{f}^{(n)}$ lies in the quadrangle with corners $\big(\frac{1}{2}, \frac{1}{4}, \frac{1}{4}\big)$, $\big(\frac{2F_h}{3F_h+2}, \frac{F_h}{3F_h+2}, \frac{2}{3F_h+2}\big)$, $\big(\frac{1}{2}, \frac{F_h-1}{2F_h}, \frac{1}{2F_h}\big)$, and \linebreak $\big(\frac{F_h+2}{3F_h+2}, \frac{F_h}{3F_h+2}, \frac{F_h}{3F_h+2}\big)$.
In particular, note that $f^{(n)}_1 - f^{(n)}_2 \ge \frac{1}{2F_h}$. 

Assume now that $i_{n-1} = 3$. (The situation is similar if $i_{n-1}, \ldots, i_{n-\ell+1}$ are alternatingly $2$ and~$1$, and $i_{n-\ell} = 3$.) 
Then $(f^{(n-1)}_1 - f^{(n-1)}_2)/f^{(n-1)}_1$ is minimal when $\mathbf{f}^{(n)} = \big(\frac{1}{2}, \frac{F_h-1}{2F_h}, \frac{1}{2F_h}\big)$, which implies that $\mathbf{f}^{(n-1)} = \big(\frac{1}{3}, \frac{F_h-1}{3F_h}, \frac{F_h+1}{3F_h}\big)$, thus $(f^{(n-1)}_1 - f^{(n-1)}_2)/f^{(n-1)}_1 \ge \frac{1}{F_h}$. 
A study of several cases shows that this is a lower bound for $(f_{j_n}^{(n)}-f_{k_n}^{(n)})/f_{i_n}^{(n)}$ when $\{i_n,i_{n+1},\ldots,i_{n+h-1}\} = \mathcal{A}$ for all $n \in \mathbb{N}$. \qed
\end{proof}

\begin{lemma} \label{l:Br2}
Let $\omega$ be a Brun  word with directive sequence $(\beta_{i_n,j_n})_{n\in\mathbb{N}}$.
Then there is a sequence of matrices $(\tilde{M}_n)_{n\in\mathbb{N}}$ satisfying~\eqref{e:tMtilde} with
\[
\big\|\tilde{M}_n\, \mathbf{e}_{i_n}\big\|_1 = 1 - \frac{f_{j_n}^{(n)}-f_{k_n}^{(n)}}{f_{i_n}^{(n)}}  \le 1 
\]
and $\tilde{M}_n\, \mathbf{e}_j = \mathbf{e}_j$ for all $j \in \mathcal{A} \setminus \{i_n\}$, $n \in \mathbb{N}$.
\end{lemma}

\begin{proof}
For each $n \in \mathbb{N}$, let $\tilde{M}_n$ be the matrix built from~$M_n$ by subtracting $\mathbf{f}^{(n)}/f_{i_n}^{(n)}$ from the $i_n$-th column.
Then 
\[
\big\|\tilde{M}_n\, \mathbf{e}_{i_n}\big\|_1 = \bigg(1 - \frac{f_{i_n}^{(n)}}{f_{i_n}^{(n)}}\bigg) + \bigg(1 - \frac{f_{j_n}^{(n)}}{f_{i_n}^{(n)}}\bigg) + \frac{f_{k_n}^{(n)}}{f_{i_n}^{(n)}} = 1 - \frac{f_{j_n}^{(n)}-f_{k_n}^{(n)}}{f_{i_n}^{(n)}},
\]
and $\tr{\!\tilde{M}_n}\, \mathbf{e}_j = \mathbf{e}_j$ for all $j \in \mathcal{A} \setminus \{i_n\}$.
Since adding a multiple of~$\mathbf{f}^{(n)}$ to a column of~$M_n$ does not change $\tr{\!M_n}\, \mathbf{x}$ for $\mathbf{x} \in (\mathbf{f}^{(n)})^\bot$, we have $\tilde{M}_n\, \mathbf{x} = M_n\, \mathbf{x}$ for all $\mathbf{x} \in (\mathbf{f}^{(n)})^\bot$.
Using that $\tr{\!M_n}\, (\mathbf{f}^{(n)})^\bot = (\mathbf{f}^{(n+1)})^\bot$, we obtain inductively that \eqref{e:tMtilde} holds, which proves the lemma. \qed
\end{proof}

\begin{theorem} \label{t:balanced2}
Let $h \in \mathbb{N}$. 
There is a constant $C(h)$ such that each Brun word over $3$~letters with strong partial quotients bounded by~$h$, i.e., with directive sequence $(\beta_{i_n j_n})_{n\in\mathbb{N}}$ satisfying $\{i_n,\ldots,i_{n+h-1}\} = \{1,2,3\}$ for all $n \in \mathbb{N}$, is $C(h)$-balanced.
\end{theorem}

\begin{proof}
The proof runs along the same lines as that of Theorem~\ref{t:balanced}.
Here, Lemma~\ref{l:Br2} implies that $\big\|\tr{(\tilde{M}_{[n,n+h)})}\big\|_\infty \le \frac{1}{F_h}$ for all $n \in \mathbb{N}$, similarly to Lemma~\ref{l:AR3}, thus $\|\tr{(\tilde{M}_{[0,n)})}\|_\infty = \mathcal{O}(F_h^{-n/h})$. 
Lemma~\ref{l:bounded} and Proposition~\ref{p:bounded3} conclude the proof. \qed
\end{proof}

\section{Balancedness of almost all words}
We use here the results of \cite{Avila-Delecroix} on Lyapunov exponents to prove that for almost all directive sequences for Brun or Arnoux-Rauzy algorithms, the associated infinite words have finite balances.
Here, we define cylinders for both algorithms as follows: given a finite word~$w$, we denote by $[w]$ the set of frequency vectors for which the continued fraction expansion starts by~$w$.

\begin{theorem} \label{t:almostall}
Let $\mu$ be an ergodic invariant probability measure for the Arnoux-Rauzy algorithm (on $d$~letters) such that $\mu([w]) > 0$ for the cylinder corresponding to a word $w_0 w_1 \cdots w_{n-1} \in \mathcal{A}^*$ with $\{w_0,w_1,\ldots,w_{n-1}\} = \mathcal{A}$.
Then, for $\mu$-almost every~$\mathbf{f}$ in the Rauzy gasket, the Arnoux-Rauzy word $\omega_\mathrm{AR}(\mathbf{f})$ is finitely balanced.

Let $\mu$ be an ergodic invariant probability measure for the Brun algorithm on $3$~letters such that $\mu([w]) > 0$ for the cylinder corresponding to a word $w = (i_0,j_0) \cdots (i_{n-1},j_{n-1})$ with $\{j_0,j_1,\ldots,j_{n-1}\} = \{1,2,3\}$.
Then, for $\mu$-almost every~$\mathbf{f}$, the Brun word $\omega_\mathrm{Br}(\mathbf{f})$ is finitely balanced.
\end{theorem}

\begin{proof}
  From \cite{Avila-Delecroix}, we know that the second Lyapunov exponent of the cocycle $M_{[0,n)}$ is negative and hence $\| \tr{(M_{[0,n)})}\, \mathbf{v}\|$ decays exponentially fast for $\mu$-almost every~$\mathbf{f}$ and all $\mathbf{v} \in \mathbf{f}^\bot$. 
By Proposition~\ref{p:bounded3}, this implies that $\omega_\mathrm{AR}(\mathbf{f})$ and $\omega_\mathrm{Br}(\mathbf{f})$ respectively are finitely balanced. \qed
\end{proof}

The Brun algorithm admits an invariant ergodic probability measure absolutely continuous with respect to Lebesgue~\cite{Arnoux-Nogueira93}.
Therefore, we have the following corollary of Theorem~\ref{t:almostall}.

\begin{corollary}
  For Lebesgue almost all frequency vectors $\mathbf{f} \in \mathbb{R}^3$, the Brun word $\omega_\mathrm{Br}(\mathbf{f})$ is finitely balanced.
\end{corollary}

\section{Imbalances in Brun sequences}

Similarly to the construction of unbalanced Arnoux-Rauzy words (over $3$~letters) in~\cite{Cassaigne-Ferenczi-Zamboni00}, we construct now unbalanced Brun words for $d=3$. 
First, for any sequence $(\sigma_n)_{n\in\mathbb N}$ satisfying \eqref{eq:Brun:admis}, define a sequence $(\tilde\sigma_n)_{n\in\mathbb N}$ as follows:
\[
	\tilde\sigma_0=\zeta_1
\quad\text{and}\quad
	\tilde\sigma_n=\begin{cases}
		\zeta_1:1\mapsto1, 2\mapsto2, 3\mapsto23, &\text{if } \sigma_{n-1}\sigma_n=\beta_{ij}\beta_{ij}
	,\\
		\zeta_2:1\mapsto1, 2\mapsto3, 3\mapsto32, &\text{if } \sigma_{n-1}\sigma_n=\beta_{ij}\beta_{ji}
	,\\
		\zeta_3:1\mapsto2, 2\mapsto3, 3\mapsto31, &\text{if } \sigma_{n-1}\sigma_n=\beta_{ij}\beta_{jk}
	,\end{cases}
\]
see Figure~\ref{fig:Brun}.
If $\omega$ and $\tilde\omega$ have the directive sequences $(\sigma_n)_{n\in\mathbb N}$ and $(\tilde\sigma_n)_{n\in\mathbb N}$ respectively, then $\omega$ and~$\tilde\omega$ differ only by a bijective letter-to-letter morphism, which does not influence the balance properties.
The proofs of the following results will be given by the end of this section.

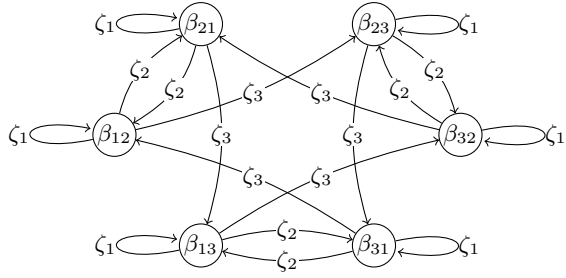
\begin{figure}[ht]
\centering
\begin{tikzpicture}[xscale=2.3, yscale=1.7]
\foreach \pos/\lbl in {0/32, 1/23, 2/21, 3/12, 4/13, 5/31}
{
	\node[draw,circle,inner sep=0.1ex] (\lbl) at (60*\pos:1) {$\beta_{\lbl}$};
}
\foreach \lba/\lbb in {21/12, 12/21, 32/23, 23/32, 13/31, 31/13}
{
	\path[->] (\lba) edge [bend left=20] node[fill=white,inner sep=0.05ex] {$\zeta_2$} (\lbb);
}
\foreach \lba/\lbb in {21/13, 13/32, 32/21}
{
	\path[->] (\lba) edge [bend left=12] node[fill=white,inner sep=0.05ex] {$\zeta_3$} (\lbb);
}
\foreach \lba/\lbb in {23/31, 31/12, 12/23}
{
	\path[->] (\lba) edge [bend right=12] node[fill=white,inner sep=0.05ex] {$\zeta_3$} (\lbb);
}
\foreach \lbl in {12,13,21}
{
	\path[->] (\lbl) edge [loop left] node[inner sep=0pt] {$\zeta_1$} (\lbl);
}
\foreach \lbl in {31,32,23}
{
	\path[->] (\lbl) edge [loop right] node[inner sep=0pt] {$\zeta_1$} (\lbl);
}
\end{tikzpicture}
\caption{Relation between the directive sequences of $\omega$ and~$\tilde\omega$.
If we follow the directive seqence of $\omega$ on the nodes, then we read the directive sequence of $\tilde\omega$ on the edges.}
\label{fig:Brun}
\end{figure}

\begin{proposition}\label{prop:notb-Brun}
Let $C\in\mathbb N$ and let
\[
\omega = \underbrace{\zeta_1^{C-1}\zeta_2\zeta_1\zeta_3^2\zeta_1}_{=\tau_{C-1}}
	\underbrace{\zeta_1^{C-2}\zeta_2\zeta_1\zeta_3^2\zeta_1}_{=\tau_{C-2}}
    \dotsm
	\underbrace{\zeta_1\zeta_2\zeta_1\zeta_3^2\zeta_1}_{=\tau_1}
	\underbrace{\zeta_2\zeta_1\zeta_3^2\zeta_1}_{=\tau_0}(\omega')
\]
for some $\omega' \in \{1,2,3\}^\mathbb{N}$ containing the letters $1$ and~$3$.
Then $\omega$ is not $C$-balanced.
\end{proposition}

Notice that the segment $\zeta_1^k\zeta_2\zeta_1\zeta_3^2\zeta_1 \, \zeta_1^{k-1}\zeta_2\zeta_1\zeta_3^2\zeta_1$
 in $(\tilde\sigma_n)_{n\in\mathbb N}$ comes from
 the segment $\beta_{ij}^k\beta_{ji}^2\beta_{ik}\beta_{kj}^2 \, \beta_{kj}^{k-1}\beta_{jk}^2\beta_{ki}\beta_{ij}^2$
 in $(\sigma_n)_{n\in\mathbb N}$.
Therefore, there exist directive sequences where each substitution $\beta_{ij}$ occurs with gaps that are bounded by~$2C+5$.

The proposition shows that for any $C$ there are uncountably many Brun words that are not $C$-balanced.
Moreover, there are also uncountably many Brun words that are not finitely balanced.

\begin{theorem}\label{thm:notb-Brun}
Let $(c_k)_{k\in\mathbb N}$ be a sequence of natural numbers such that
\[
	c_k>12\sqrt3\, 3^{N(c_0)+N(c_1)+\cdots+N(c_{k-1})} k
\quad\text{for all}\quad k\in\mathbb{N},
\]
with $N(c) = c(c+1)/2+3c$. Let $\rho_c=\tau_{c-1}\tau_{c-2}\dotsm \tau_1\tau_0$, with $\tau_j$ as in Proposition~\ref{prop:notb-Brun}.
Then the Brun word with directive sequence $\rho_{c_0}\rho_{c_1}\dotsm$ is not finitely balanced.
\end{theorem}

To prove these statements, we will use techniques that are typical for finding imbalances in $S$-adic sequences.
Let $u,v\in\mathcal{A}^*$. Then we put $\mathbf\Delta_{u,v}=\Ab(u) - \Ab(v)$.
For any substitution~$\sigma$, we clearly have
\begin{equation}\label{eq:sigma-Delta}
	\mathbf\Delta_{\sigma(u),\sigma(v)}=M_\sigma\mathbf\Delta_{u,v}
,\end{equation}
and consequently:

\begin{lemma}\label{lem:sigma-Delta}
Let $\sigma$ be a substitution over the alphabet~$\mathcal{A}$ such that the images of all letters under $\sigma$ start with the same letter $a \in \mathcal{A}$.
Let $u,v$ be non-empty factors of a word $\omega \in \mathcal{A}^\mathbb{N}$. 
Then $\sigma(\omega)$ contains factors $u', v'$ with $\mathbf\Delta'=\mathbf\Delta_{u', v'}=M_\sigma\mathbf\Delta+p\,\mathbf e_a$
 for all $p\in\{0,\pm1,\pm2\}$.
\end{lemma}

\begin{proof}[of Proposition~\ref{prop:notb-Brun}]
Consider a pair of words $u, v$, with $\mathbf\Delta_{u,v}=\tr{(q+1,-q,-1)}$.
Then, using \eqref{eq:sigma-Delta} and applying Lemma~\ref{lem:sigma-Delta} with the substitution $\zeta_3^2: 1\mapsto3,\, 2\mapsto31,\, 3\mapsto312$, we obtain the following chain of $\mathbf\Delta$'s:
\[
	\vect{q+1}{-q}{-1} \xrightarrow{\zeta_1} \vect{q+1}{-q-1}{-1} \xrightarrow[p=2]{\zeta_3^2} \vect{-q-2}{-1}{1}
	\xrightarrow{\zeta_2\zeta_1} \vect{q-2}{1}{1} \xrightarrow{\zeta_1^q} \vect{-(q+2)}{q+1}{1},
\]
and by symmetry $\tr{(-q-1,q,1)}\xrightarrow{\tau_q}\tr{(q-2,-q-1,-1)}$.

ting from the pair of factors $1,3$ of~$\omega'$, we have the chain
\[
\vect{1}{0}{-1} \xrightarrow{\tau_0} \vect{-2}{1}{1}
    \rightarrow\cdots\rightarrow \pm\vect{-C}{C-1}{1} \xrightarrow{\tau_{C-1}} \pm\vect{C+1}{-C}{-1}
.\]
The last vector sums to zero, therefore it corresponds to factors $u,v$ of~$\omega$ such that $|u|=|v|$ and $\bigl||u|_1-|v|_1\bigr|=C+1$. \qed
\end{proof}

\begin{lemma}\label{lem:12k3}
Let $\omega$ be a Brun word and let $C,N\in\mathbb N$ be such that $\omega^{(N)}$ is not $\bigl(12\sqrt3\,3^N C\bigr)$-balanced.
Then $\omega$ is not $C$-balanced. 
\end{lemma}

\begin{proof}
We will only sketch the proof.
According to Lemma~\ref{l:bounded}, there is a prefix~$u$ of~$\omega^{(N)}$ such that $\|\pi_N\, \mathbf{x}\|_\infty > \frac14 12 \sqrt3\, 3^N C$, where $\mathbf{x} = \Ab(u)$ and  $\pi_N$ is the projection along the frequency vector $\mathbf{f}^{(N)}$ of $\omega^{(N)}$ onto~$\mathbf{1}^\bot$.
Then the frequency vector of~$\omega$ is $\mathbf{f} = M_{[0,N)}\, \mathbf{f}^{(N)}$, and $M_{[0,N)}\, \mathbf{x}$ is the abelianization of a prefix of~$\omega$.

Let $\gamma$ be the angle between the vectors $\mathbf{f}^{(N)}$ and~$\mathbf{x}$.
Then it can be verified that applying $M_n$ divides the angle between two non-negative vectors by at most~$3$, thus the angle between $\mathbf{f}$ and~$\mathbf{x}$ is at least $\gamma/3^N$.
Since the matrices $M_n$ are of the form identity matrix + non-negative matrix, and the vector $\mathbf{x}$ is non-negative,  we get that $\|\mathbf{x}\|_2 \geq \|\mathbf{x}^{(N)}\|_2$.
Therefore the (orthogonal) distance~$\delta$ of the point~$\mathbf{x}$ from the line $\mathbb{R}\,\mathbf{f}$ is at least $1/3^N$ times the distance of $\mathbf{x}$ from $\mathbb{R}\, \mathbf{f}^{(N)}$, which is at least $\frac{1}{\sqrt3} \|\pi_N\, \mathbf{x}\|_2$.
Altogether, $\delta \geq \frac{1}{\sqrt3\cdot3^N} \|\pi_N\, \mathbf{x}\|_2 \geq \frac{1}{3\cdot3^N} \|\pi_N\, \mathbf{x}\|_\infty > 3C$.

Finally, $\|\pi\, M_{[0,N)}\, \mathbf{x}\|_\infty \geq \frac{1}{\sqrt3} \|\pi\, M_{[0,N)}\, \mathbf{x}\|_2 \geq \frac13 \delta > C$, which means according to Lemma~\ref{l:bounded} that $\omega$ is not $C$-balanced. \qed
\end{proof}

\begin{proof}[of Theorem~\ref{thm:notb-Brun}]
The Brun word with directive sequence $\rho_{c_k}\rho_{c_{k+1}}\rho_{c_{k+1}}\dotsm$ is not $N(c_k)$-balanced according
 to Proposition \ref{prop:notb-Brun}.
Lemma \ref{lem:12k3} therefore gives that $\omega$ is not $k$-balanced. \qed
\end{proof}

\subsubsection*{Acknowledgements}
The authors are very grateful to Val\'erie Berth\'e, Timo Jolivet, and S\'ebastien Labb\'e for many fruitful discussions. 

This work was supported by the Grant Agency of the Czech Technical University in Prague grant SGS11/162/OHK4/3T/14, Czech Science Foundation grant 13-03538S, and the ANR/FWF project ``FAN -- Fractals and Numeration'' (ANR-12-IS01-0002,
FWF grant I1136).  

\bibliographystyle{splncs03}
\bibliography{balances}
\end{document}